\documentclass[conference, twocolumn]{IEEEtran}
\usepackage{amssymb, amsmath}
\usepackage{amsfonts}
\usepackage{bbm}
\usepackage{mathrsfs}
\usepackage{xspace}
\usepackage{bm}

\usepackage{cite}
\usepackage{epsfig}
\usepackage{cases}
\usepackage{psfrag}
\usepackage{enumerate}

\usepackage{makeidx}
\usepackage{supertabular}
\usepackage{acronym}

\newtheorem{theorem}{Theorem}

\newtheorem{proposition}{Proposition}

\newtheorem{notation}{Notation}

\newenvironment{textbmatrix}{   \setlength{\arraycolsep}{2.5pt}%
                                                                \big[\begin{matrix}}{\end{matrix}\big]%
                                                                \raisebox{0.08ex}{\vphantom{M}}}

\def\be{\begin{equation}}
\def\ee{\end{equation}}
\def\een{\nonumber \end{equation}}
\def\mat{\begin{bmatrix}}
\def\emat{\end{bmatrix}}
\def\btm{\begin{textbmatrix}}
\def\etm{\end{textbmatrix}}

\def\ba#1\ea{\begin{align}#1\end{align}}
\def\bs#1\es{\begin{split}#1\end{split}}
\def\bg#1\eg{\begin{gather}#1\end{gather}}
\def\bi#1\ei{\begin{itemize}#1\end{itemize}}

\newcommand{\safemath}[2]{\newcommand{#1}{\ensuremath{#2}\xspace}}

\newcommand{\lefto}{\mathopen{}\left}

\DeclareMathOperator{\Varop}{\mathbb{V}\!\mathrm{ar}} 

\safemath{\interior}{\mathrm{Int}}                       

\safemath{\dfn}{:=}                                                     
\safemath{\dirac}{\delta}                                       

\safemath{\No}{N_0}                                                     
\safemath{\Es}{E_s}                                                     
\safemath{\Eb}{E_b}                                                     
\safemath{\EbNo}{\frac{\Eb}{\No}}
\safemath{\EsNo}{\frac{\Es}{\No}}

\DeclareMathOperator{\CHop}{\ensuremath{\mathbb{H}}} 
\safemath{\tvir}{h_{\CHop}}                                     
\safemath{\tvtf}{L_{\CHop}}                                     
\safemath{\spf}{S_{\CHop}}
\safemath{\bff}{H_{\CHop}}                                      

\safemath{\ircf}{R_{h}}                                         
\safemath{\scf}{R_{S}}                                          
\safemath{\tfcf}{R_{L}}                                         
\safemath{\bfcf}{R_{H}}                                         

\safemath{\mi}{I}                                                       
\safemath{\capacity}{C}                                         

\newcommand{\pdf[2]}{f_{{#1}}\left({#2}\right)}                    

\safemath{\uniform}{\mathcal{U}}                        
\safemath{\normal}{\mathcal{N}}                         
\safemath{\circnorm}{\mathcal{CN}}                      
\safemath{\mchain}{\leftrightarrow}                     

\safemath{\dB}{\,\mathrm{dB}}
\safemath{\dBm}{\,\mathrm{dBm}}
\safemath{\Hz}{\,\mathrm{Hz}}
\safemath{\kHz}{\,\mathrm{kHz}}
\safemath{\MHz}{\,\mathrm{MHz}}
\safemath{\GHz}{\,\mathrm{GHz}}
\safemath{\s}{\,\mathrm{s}}
\safemath{\ms}{\,\mathrm{ms}}
\safemath{\mus}{\,\mathrm{\mu s}}
\safemath{\ns}{\,\mathrm{ns}}
\safemath{\meter}{\,\mathrm{m}}
\safemath{\mm}{\,\mathrm{mm}}
\safemath{\cm}{\,\mathrm{cm}}
\safemath{\m}{\,\mathrm{m}}
\safemath{\W}{\,\mathrm{W}}
\safemath{\J}{\,\mathrm{J}}
\safemath{\K}{\,\mathrm{K}}
\safemath{\bit}{\,\mathrm{bit}}
\safemath{\mW}{\,\mathrm{mW}}
\safemath{\nW}{\,\mathrm{nW}}
\safemath{\pW}{\,\mathrm{pW}}
\safemath{\muW}{\,\mu\mathrm{W}}
\safemath{\Watt}{\,\mathrm{W}}
\safemath{\kbps}{\,\mathrm{kb/s}}
\safemath{\Mbps}{\,\mathrm{Mb/s}}
\safemath{\bpsHz}{\,\mathrm{b/s/Hz}}

\safemath{\define}{\triangleq}                  

\safemath{\equivalent}{\sim}
\safemath{\distas}{\sim}                                        

\safemath{\reals}{\mathbb{R}}
\safemath{\positivereals}{\mathbb{R}^{+}}
\safemath{\integers}{\mathbb{Z}}
\safemath{\posint}{\mathbb{Z}_{+}}
\safemath{\naturals}{\mathbb{N}}
\safemath{\complexset}{\mathbb{C}}

\safemath{\setA}{\mathcal{A}}
\safemath{\setB}{\mathcal{B}}
\safemath{\setC}{\mathcal{C}}
\safemath{\setD}{\mathcal{D}}
\safemath{\setE}{\mathcal{E}}
\safemath{\setF}{\mathcal{F}}
\safemath{\setG}{\mathcal{G}}
\safemath{\setH}{\mathcal{H}}
\safemath{\setI}{\mathcal{I}}
\safemath{\setJ}{\mathcal{J}}
\safemath{\setK}{\mathcal{K}}
\safemath{\setL}{\mathcal{L}}
\safemath{\setM}{\mathcal{M}}
\safemath{\setN}{\mathcal{N}}
\safemath{\setO}{\mathcal{O}}
\safemath{\setP}{\mathcal{P}}
\safemath{\setQ}{\mathcal{Q}}
\safemath{\setR}{\mathcal{R}}
\safemath{\setS}{\mathcal{S}}
\safemath{\setT}{\mathcal{T}}
\safemath{\setU}{\mathcal{U}}
\safemath{\setV}{\mathcal{V}}
\safemath{\setW}{\mathcal{W}}
\safemath{\setX}{\mathcal{X}}
\safemath{\setY}{\mathcal{Y}}
\safemath{\setZ}{\mathcal{Z}}
\safemath{\emptySet}{\varnothing}

\safemath{\bma}{\mathbf{a}}
\safemath{\bmb}{\mathbf{b}}
\safemath{\bmc}{\mathbf{c}}
\safemath{\bmd}{\mathbf{d}}
\safemath{\bme}{\mathbf{e}}
\safemath{\bmf}{\mathbf{f}}
\safemath{\bmg}{\mathbf{g}}
\safemath{\bmh}{\mathbf{h}}
\safemath{\bmi}{\mathbf{i}}
\safemath{\bmj}{\mathbf{j}}
\safemath{\bmk}{\mathbf{k}}
\safemath{\bml}{\mathbf{l}}
\safemath{\bmm}{\mathbf{m}}
\safemath{\bmn}{\mathbf{n}}
\safemath{\bmo}{\mathbf{o}}
\safemath{\bmp}{\mathbf{p}}
\safemath{\bmq}{\mathbf{q}}
\safemath{\bmr}{\mathbf{r}}
\safemath{\bms}{\mathbf{s}}
\safemath{\bmt}{\mathbf{t}}
\safemath{\bmu}{\mathbf{u}}
\safemath{\bmv}{\mathbf{v}}
\safemath{\bmw}{\mathbf{w}}
\safemath{\bmx}{\mathbf{x}}
\safemath{\bmy}{\mathbf{y}}
\safemath{\bmz}{\mathbf{z}}

\bmdefine{\biad}{a}
\bmdefine{\bibd}{b}
\bmdefine{\bicd}{c}
\bmdefine{\bidd}{d}
\bmdefine{\bied}{e}
\bmdefine{\bifd}{f}
\bmdefine{\bigd}{g}
\bmdefine{\bihd}{h}
\bmdefine{\biid}{i}
\bmdefine{\bijd}{j}
\bmdefine{\bikd}{k}
\bmdefine{\bild}{l}
\bmdefine{\bimd}{m}
\bmdefine{\bind}{n}
\bmdefine{\biod}{o}
\bmdefine{\bipd}{p}
\bmdefine{\biqd}{q}
\bmdefine{\bird}{r}
\bmdefine{\bisd}{s}
\bmdefine{\bitd}{t}
\bmdefine{\biud}{u}
\bmdefine{\bivd}{v}
\bmdefine{\biwd}{w}
\bmdefine{\bixd}{x}
\bmdefine{\biyd}{y}
\bmdefine{\bizd}{z}

\bmdefine{\bixid}{\xi}
\bmdefine{\bilambdad}{\lambda}
\bmdefine{\bimud}{\mu}
\bmdefine{\bithetad}{\theta}
\bmdefine{\biphid}{\phi}

\safemath{\bmia}{\biad}
\safemath{\bmib}{\bibd}
\safemath{\bmic}{\bicd}
\safemath{\bmid}{\bidd}
\safemath{\bmie}{\bied}
\safemath{\bmif}{\bifd}
\safemath{\bmig}{\bigd}
\safemath{\bmih}{\bihd}
\safemath{\bmii}{\biid}
\safemath{\bmij}{\bijd}
\safemath{\bmik}{\bikd}
\safemath{\bmil}{\bild}
\safemath{\bmim}{\bimd}
\safemath{\bmin}{\bind}
\safemath{\bmio}{\biod}
\safemath{\bmip}{\bipd}
\safemath{\bmiq}{\biqd}
\safemath{\bmir}{\bird}
\safemath{\bmis}{\bisd}
\safemath{\bmit}{\bitd}
\safemath{\bmiu}{\biud}
\safemath{\bmiv}{\bivd}
\safemath{\bmiw}{\biwd}
\safemath{\bmix}{\bixd}
\safemath{\bmiy}{\biyd}
\safemath{\bmiz}{\bizd}

\safemath{\bmxi}{\bixid}
\safemath{\bmlambda}{\bilambdad}
\safemath{\bmmu}{\bimud}
\safemath{\bmtheta}{\bithetad}
\safemath{\bmphi}{\biphid}

\safemath{\bA}{\mathbf{A}}
\safemath{\bB}{\mathbf{B}}
\safemath{\bC}{\mathbf{C}}
\safemath{\bD}{\mathbf{D}}
\safemath{\bE}{\mathbf{E}}
\safemath{\bF}{\mathbf{F}}
\safemath{\bG}{\mathbf{G}}
\safemath{\bH}{\mathbf{H}}
\safemath{\bI}{\mathbf{I}}
\safemath{\bJ}{\mathbf{J}}
\safemath{\bK}{\mathbf{K}}
\safemath{\bL}{\mathbf{L}}
\safemath{\bM}{\mathbf{M}}
\safemath{\bN}{\mathbf{N}}
\safemath{\bO}{\mathbf{O}}
\safemath{\bP}{\mathbf{P}}
\safemath{\bQ}{\mathbf{Q}}
\safemath{\bR}{\mathbf{R}}
\safemath{\bS}{\mathbf{S}}
\safemath{\bT}{\mathbf{T}}
\safemath{\bU}{\mathbf{U}}
\safemath{\bV}{\mathbf{V}}
\safemath{\bW}{\mathbf{W}}
\safemath{\bX}{\mathbf{X}}
\safemath{\bY}{\mathbf{Y}}
\safemath{\bZ}{\mathbf{Z}}

\bmdefine{\biAd}{A}
\bmdefine{\biBd}{B}
\bmdefine{\biCd}{C}
\bmdefine{\biDd}{D}
\bmdefine{\biEd}{E}
\bmdefine{\biFd}{F}
\bmdefine{\biGd}{G}
\bmdefine{\biHd}{H}
\bmdefine{\biId}{I}
\bmdefine{\biJd}{J}
\bmdefine{\biKd}{K}
\bmdefine{\biLd}{L}
\bmdefine{\biMd}{M}
\bmdefine{\biOd}{N}
\bmdefine{\biPd}{O}
\bmdefine{\biQd}{P}
\bmdefine{\biRd}{R}
\bmdefine{\biSd}{S}
\bmdefine{\biTd}{T}
\bmdefine{\biUd}{U}
\bmdefine{\biVd}{V}
\bmdefine{\biWd}{W}
\bmdefine{\biXd}{X}
\bmdefine{\biYd}{Y}
\bmdefine{\biZd}{Z}

\bmdefine{\biDelta}{\Delta}
\bmdefine{\biLambda}{\Lambda}
\bmdefine{\biPhi}{\Phi}
\bmdefine{\biSigma}{\Sigma}
\bmdefine{\biOmega}{\Omega}
\bmdefine{\biTheta}{\Theta}

\safemath{\bimA}{\biAd}
\safemath{\bimB}{\biBd}
\safemath{\bimC}{\biCd}
\safemath{\bimD}{\biDd}
\safemath{\bimE}{\biEd}
\safemath{\bimF}{\biFd}
\safemath{\bimG}{\biGd}
\safemath{\bimH}{\biHd}
\safemath{\bimI}{\biId}
\safemath{\bimJ}{\biJd}
\safemath{\bimK}{\biKd}
\safemath{\bimL}{\biLd}
\safemath{\bimM}{\biMd}
\safemath{\bimN}{\biNd}
\safemath{\bimO}{\biOd}
\safemath{\bimP}{\biPd}
\safemath{\bimQ}{\biQd}
\safemath{\bimR}{\biRd}
\safemath{\bimS}{\biSd}
\safemath{\bimT}{\biTd}
\safemath{\bimU}{\biUd}
\safemath{\bimV}{\biVd}
\safemath{\bimW}{\biWd}
\safemath{\bimX}{\biXd}
\safemath{\bimY}{\biYd}
\safemath{\bimZ}{\biZd}

\safemath{\bDelta}{\bielta}
\safemath{\bLambda}{\biLambda}
\safemath{\bPhi}{\biPhi}
\safemath{\bSigma}{\biSigma}
\safemath{\bOmega}{\biOmega}
\safemath{\bTheta}{\biTheta}

\safemath{\veca}{\bma}
\safemath{\vecb}{\bmb}
\safemath{\vecc}{\bmc}
\safemath{\vecd}{\bmd}
\safemath{\vece}{\bme}
\safemath{\vecf}{\bmf}
\safemath{\vecg}{\bmg}
\safemath{\vech}{\bmh}
\safemath{\veci}{\bmi}
\safemath{\vecj}{\bmj}
\safemath{\veck}{\bmk}
\safemath{\vecl}{\bml}
\safemath{\vecm}{\bmm}
\safemath{\vecn}{\bmn}
\safemath{\veco}{\bmo}
\safemath{\vecp}{\bmp}
\safemath{\vecq}{\bmq}
\safemath{\vecr}{\bmr}
\safemath{\vecs}{\bms}
\safemath{\vect}{\bmt}
\safemath{\vecu}{\bmu}
\safemath{\vecv}{\bmv}
\safemath{\vecw}{\bmw}
\safemath{\vecx}{\bmx}
\safemath{\vecy}{\bmy}
\safemath{\vecz}{\bmz}
\safemath{\vecZero}{\bZero}

\safemath{\vecxi}{\bmxi}
\safemath{\veclambda}{\bmlambda}
\safemath{\vecmu}{\bmmu}
\safemath{\vectheta}{\bmtheta}
\safemath{\vecphi}{\bmphi}

\safemath{\matA}{\bA}
\safemath{\matB}{\bB}
\safemath{\matC}{\bC}
\safemath{\matD}{\bD}
\safemath{\matE}{\bE}
\safemath{\matF}{\bF}
\safemath{\matG}{\bG}
\safemath{\matH}{\bH}
\safemath{\matI}{\bI}
\safemath{\matJ}{\bJ}
\safemath{\matK}{\bK}
\safemath{\matL}{\bL}
\safemath{\matM}{\bM}
\safemath{\matN}{\bN}
\safemath{\matO}{\bO}
\safemath{\matP}{\bP}
\safemath{\matQ}{\bQ}
\safemath{\matR}{\bR}
\safemath{\matS}{\bS}
\safemath{\matT}{\bT}
\safemath{\matU}{\bU}
\safemath{\matV}{\bV}
\safemath{\matW}{\bW}
\safemath{\matX}{\bX}
\safemath{\matY}{\bY}
\safemath{\matZ}{\bZ}
\safemath{\matZero}{\bZero}

\safemath{\matDelta}{\bDelta}
\safemath{\matLambda}{\bLambda}
\safemath{\matPhi}{\bPhi}
\safemath{\matSigma}{\bSigma}
\safemath{\matOmega}{\bOmega}
\safemath{\matTheta}{\bTheta}

\safemath{\matIdentity}{\matI}

\usepackage{fancyhdr,graphicx,color,listings}

\newcommand{\sectionname}{Sec.}

\renewcommand{\figurename}{Fig.}

\renewcommand{\tablename}{Tab.}

\newcommand{\theoremname}{Theorem}
\newcommand{\equationname}{Eq.}
\newcommand{\equationsname}{Eqs.}
\newcommand{\propname}{Prop.}

\newcommand{\covs}{\mathbf{\bSigma}}

\safemath{\sourceindex}{s}
\safemath{\othersource}{s'}
\safemath{\sourceset}{\setS}
\safemath{\sourcesubset}{\setL}
\safemath{\sourcenumber}{S}

\safemath{\destindex}{d}

\safemath{\relayindex}{r}
\safemath{\otherrelay}{r'}
\safemath{\relayset}{\setR}
\safemath{\relaysubset}{\setT}
\safemath{\relaynumber}{R}

\safemath{\setname}{\setA}

\newcommand{\relayfunction[2]}{f_{#1}^{#2}}

\safemath{\timeindex}{t}

\safemath{\powerSymbol}{p}

\newcommand{\powermax[1]}{\overline{\powerSymbol}_{#1}}

\safemath{\upperboundSymbol}{\mathcal{C}}
\newcommand{\upperbound[2]}{\displaystyle\upperboundSymbol_{#1}{#2}}
\safemath{\capacitySymbol}{C}
\safemath{\SNRSymbol}{\gamma}
\newcommand{\SNR[2]}{\SNRSymbol_{{#1}{#2}}}

\safemath{\pathgainSymbol}{h}
\newcommand{\pathgain[2]}{\pathgainSymbol_{{#1}{#2}}}
\safemath{\noiseSymbol}{Z}
\newcommand{\noise[1]}{\noiseSymbol_{#1}}
\safemath{\noisePower}{\sigma^2_w}
\newcommand{\AFconstant[1]}{\alpha_{#1}}

\safemath{\entropy}{\mathrm{H}}
\safemath{\diffentropy}{h}

\newcommand{\information}{\mathrm{I}}
\newcommand{\expectation[1]}{\Exop\{#1\}}

\safemath{\corrsrSymbol}{\rho}
\newcommand{\corrsr[2]}{\corrsrSymbol_{{#1}{#2}}}
\safemath{\corrrrSymbol}{\delta}
\newcommand{\corrrr[2]}{\corrrrSymbol_{{#1}{#2}}}
\safemath{\corrssSymbol}{\lambda}

\safemath{\varSpaceSymbol}{\mathcal{X}}
\newcommand{\varSpace[1]}{\varSpaceSymbol_{#1}}

\safemath{\auxiliaryvarSpace}{\mathcal{Z}}

\safemath{\othervarSpaceSymbol}{\mathcal{Y}}
\newcommand{\othervarSpace[1]}{\othervarSpaceSymbol_{#1}}
\safemath{\estimationvarSpaceSymbol}{\widehat{\mathcal{Y}}}

\safemath{\varSymbol}{X}
\safemath{\othervarSymbol}{Y}
\safemath{\auxvarSymbol}{\widehat{Y}}
\safemath{\auxiliaryvar}{Z}

\safemath{\estimationvarSymbol}{\widehat{\othervarSymbol}}
\newcommand{\var[1]}{\varSymbol_{#1}}
\newcommand{\othervar[1]}{\othervarSymbol_{#1}}

\safemath{\auxvarsampleSymbol}{z}

\safemath{\varsetSymbol}{\boldsymbol{X}}
\safemath{\othervarsetSymbol}{\boldsymbol{Y}}
\safemath{\estimationvarsetSymbol}{\widehat{\othervarsetSymbol}}
\newcommand{\varset[1]}{\varsetSymbol_{#1}}
\newcommand{\othervarset[1]}{\othervarsetSymbol_{#1}}

\safemath{\varsampleSymbol}{x}
\safemath{\othervarsampleSymbol}{y}
\safemath{\auxiliaryvarsample}{z}

\safemath{\estimationvarsampleSymbol}{\hat{\othervarsampleSymbol}}

\newcommand{\varsample[1]}{\varsampleSymbol_{#1}}
\newcommand{\othervarsample[1]}{\othervarsampleSymbol_{#1}}

\safemath{\varsetsampleSymbol}{\boldsymbol{x}}
\newcommand{\varsetsample[1]}{\varsetsampleSymbol_{#1}}
\safemath{\othervarsetsampleSymbol}{\boldsymbol{y}}
\newcommand{\othervarsetsample[1]}{\othervarsetsampleSymbol_{#1}}
\safemath{\estimationvarsetsampleSymbol}{\hat{\othervarsampleSymbol}}

\safemath{\varSequence}{X}
\safemath{\estimationvarSequence}{\widehat{X}}
\safemath{\othervarSequence}{Y}
\safemath{\auxiliaryvarSequence}{Z}

\safemath{\distortion}{d}
\safemath{\distortionLimit}{D}

\safemath{\rateSymbol}{R}
\safemath{\estimationrateSymbol}{\widehat{R}}
\safemath{\assistrateSymbol}{\tilde{R}}

\newcommand{\rate[1]}{\rateSymbol_{#1}}

\safemath{\blockindex}{b}
\safemath{\blocknumber}{B}

\safemath{\signalsetSymbol}{\mathcal{W}}
\newcommand{\signalset[1]}{\signalsetSymbol_{#1}}

\safemath{\signalnumber}{n}
\safemath{\signalSymbol}{w}
\safemath{\signalsSymbol}{\boldsymbol{w}}
\safemath{\estimationparamSymbol}{k}
\safemath{\assistparamSymbol}{z}

\newcommand{\signal[2]}{\signalSymbol_{#1}^{#2}}

\newcommand{\destsignal[2]}{\hat{\signalSymbol}_{#1}^{#2}}

\safemath{\gaussian}{\mathcal{N}}

\newcommand{\etal}{\emph{et al.\ }}

\begin{document}

\title{\LARGE{Upper{-}Bounding the Capacity of Relay Communications{ - }Part II}}

\author{
  \IEEEauthorblockN{Farshad Shams}
  \IEEEauthorblockA{Dep. of Computer Science and Engineering\\
    IMT Institute for Advanced Studies Lucca, Italy\\
    Email: f.shams@imtlucca.it}
  \and
  \IEEEauthorblockN{Marco Luise}
  \IEEEauthorblockA{Dipartimento di Ingegneria dell'Informazione\\
    University of Pisa, Italy\\
    Email: marco.luise@iet.unipi.it}
}
\maketitle

\begin{abstract}
This paper focuses on the capacity of peer{-}to{-}peer relay communications wherein the transmitter are assisted by an arbitrary number of parallel relays, i.e. there is no link and cooperation between the relays themselves.
We detail the mathematical model of different relaying strategies including cutset and amplify and forward strategies.
The cutset upper bound capacity is presented as a reference to compare another realistic strategy.
We present its outer region capacity which is lower than that in the existing literature. We show that a multiple parallel relayed network achieves its maximum capacity by virtue of only one relay or by virtue of all relays together. Adding a relay may even decrease the overall capacity or may do not change it.
We exemplify various outer region capacities of the addressed strategies with two different case studies. The results exhibit that in low signal{-}to{-}noise ratio (SNR) environments the cutset outperforms the amplify and forward strategy and this is contrary in high SNR environments.
\end{abstract}

\section{Introduction}

Wireless channels suffer from many factors, e.g. shadowing, fading, co{-}channel interference and adjacent interference that decrease the performance of the communication in terms of capacity. Cooperative communications are recently proposed as techniques to increase the system capacity and diversity gain in wireless networks.
Multi{-}hop relaying networks allow wireless terminals to aid a transmission of information between a transmitter and a corresponding destination. Such a relaying can be applied in ad{-}hoc networks to realize a transmission between two far wireless nodes in order to increase the coverage and the throughput in terms of energy consumption and capacity \cite{zahedi-thesis}. The deployment of relaying communications realize a cost effective communications by minimizing the requirements for fixed infrastructures.

The most well{-}known relaying schemes are \emph{amplify and forward} (AF), \emph{decode and forward} (DF) and \emph{compress and forward} (CF) \cite{GastparKramerGupta02}. In AF strategy, the relay nodes do not regenerate a new code word, or equivalently there is no decoder or encoder at the relays.
On the other hand, in DF and CF schemes, the relay nodes regenerate a new code word and there are encoder and decoder at each relay node.
The AF scheme is less complex than that DF and CF and for this much contributions are focused on that. The max{-}flow min{-}cut (cutset) theorem introduces the most general coding strategy and its upper bound capacity is larger than that of DF and CF strategies \cite{GastparKramerGupta02}.
Although the performance of relaying wireless communications with AWGN channels is extensively studied in terms of error probability and power allocation, but there are few remarkable contributions around information theoretic aspects of point{-}to{-}point multiple relayed communications. 

Aref in his Ph.D. dissertation \cite{aref-thesis} formulated the cutset upper bound for the case of point{-}to{-}point connection through multiple relays.
Zhang in \cite{Zhang88} established the capacity of the relay channel when the channel from the relay to
the destination is a noiseless channel of fixed capacity allowing the relay sending additional information to the destination.
Gastpar \etal in \cite{GastparKramerGupta02} introduce information theoretic aspects of a point{-}to{-}point connection using multiple parallel relays wherein each transceiver has multiple antennas.

In this paper, we treat the upper bound capacity of a peer{-}to{-}peer communication with multiple ``parallel" relays wherein there is no cooperation among relays. We investigate the upper bound capacity of point{-}to{-}point cooperative communication assisted by more than one relay in particular by a bank of \relaynumber ``parallel" relays. As the number of relays increases, more radio resources and more degrees of freedom can be jointly utilized to assist the source transmission.
We detail the mathematical model of different relaying schemes and introduce a capacity upper bound of each of them. We will show that in such a reliable network, the upper bound capacity is achieved by either only one relay or by all relays together.
We present the information{-}theoretic perspective, and outer range capacity of different relayed communication scenarios, using the cutset theorem , and also applying known AF relaying protocol. The main focus of this paper is to derive the highest achievable data rate of the transmitter, applying different strategies at the relay nodes.
We will show that increasing the number of relays does not necessarily raise the upper bound of the capacity region, and may even decrease data rate and harmfully increase the network cost in terms of power expenditure.


The rest of this paper is structured as follows. First, in \sectionname~\ref{sec:model1M1} we describe the channel model of a multiple parallel relayed point{-}to{-}point communication. The cutset theorem is used in \sectionname~\ref{sec:cutset1M1}.  
The following section is devoted to AF strategy. We exemplify the results in \sectionname~\ref{sec:casestudy1M1}. Finally, we conclude in \sectionname~\ref{sec:conclusion1M1}.

\begin{notation}\label{nota:scaledpathgain}
We use the same notation as that in \cite{Shams13}.
\end{notation}


\section{System model}\label{sec:model1M1}
We study a cooperative network that consists of one source $\sourceindex$, one destination $\destindex$, and an arbitrary number of parallel relays belonging to set $\relayset=[1,\dots,\relayindex,\dots,\relaynumber]$. We assume there is no link between relays and the communication is full{-}duplex mode, so that relay assisted transmissions must be conducted over one phase and all channels are always busy. The cooperative network of \figurename~\ref{fig:onesmultiplerencdec} consists of $2\relaynumber+2$ alphabet spaces: $\varSpace[\sourceindex]$ at the transmitter's encoder, $\varSpace[\relayindex]$ and $\othervarSpace[\relayindex]$ at each relay $\relayindex$, and finally $\othervarSpace[\destindex]$ at the destination's decoder.

\begin{figure}
  \begin{center}
    \psfrag{s}[r][b][0.85]{$\sourceindex$}
    \psfrag{r1}[c][b][0.8]{$1$}
    \psfrag{r2}[c][b][0.8]{$2$}
    \psfrag{r}[c][b][0.8]{$\relayindex$}
    \psfrag{rR}[c][b][0.8]{$\relaynumber$}
    \psfrag{d}[l][b][0.85]{$\destindex$}
    \psfrag{Enc}[c][b][0.7]{Encoder}
    \psfrag{Dec}[c][b][0.7]{Decoder}
    \psfrag{Z1}[c][b][0.75][48]{$\noise[1]$}
    \psfrag{Z2}[c][c][0.75][48]{$\noise[2]$}
    \psfrag{Zr}[c][b][0.75][-20]{$\noise[\relayindex]$}
    \psfrag{ZR}[c][b][0.75][-36]{$\noise[\relaynumber]$}
    \psfrag{Zd}[c][b][0.75]{$\noise[\destindex]$}
    \psfrag{hs1}[c][c][0.75][45]{$\sqrt{\pathgain[\sourceindex]{1}}$}
    \psfrag{hs2}[c][c][0.75][30]{$\sqrt{\pathgain[\sourceindex]{2}}$}
    \psfrag{hsr}[c][c][0.75][-14]{$\sqrt{\pathgain[\sourceindex]{\relayindex}}$}
    \psfrag{hsR}[c][c][0.75][-38]{$\sqrt{\pathgain[\sourceindex]{\relaynumber}}$}
    \psfrag{h1d}[c][c][0.75][-35]{$\sqrt{\pathgain[1]{\destindex}}$}
    \psfrag{h2d}[c][c][0.75][-17]{$\sqrt{\pathgain[2]{\destindex}}$}
    \psfrag{hrd}[c][c][0.75][25]{$\sqrt{\pathgain[\relayindex]{\destindex}}$}
    \psfrag{hRd}[c][c][0.75][41]{$\sqrt{\pathgain[\relaynumber]{\destindex}}$}
    \psfrag{hsd}[c][b][0.75]{$\sqrt{\pathgain[\sourceindex]{\destindex}}$}
    \psfrag{xs}[c][c][0.8]{$\var[\sourceindex]$}
    \psfrag{y1}[c][b][0.8]{$\othervar[1]$}
    \psfrag{x1}[c][b][0.8]{$\var[1]$}
    \psfrag{y2}[c][b][0.8]{$\othervar[2]$}
    \psfrag{x2}[c][b][0.8]{$\var[2]$}
    \psfrag{yr}[c][c][0.8]{$\othervar[\relayindex]$}
    \psfrag{xr}[c][c][0.8]{$\var[\relayindex]$}
    \psfrag{yR}[c][c][0.8]{$\othervar[\relaynumber]$}
    \psfrag{xR}[c][c][0.8]{$\var[\relaynumber]$}
    \psfrag{yd}[c][b][0.8]{$\othervar[\destindex]$}
    \psfrag{ws}[c][b][0.8]{$\signal[\sourceindex]{}$}
    \psfrag{wd}[c][b][0.8]{$\destsignal[\sourceindex]{}$}
    \includegraphics[width=0.95\columnwidth]{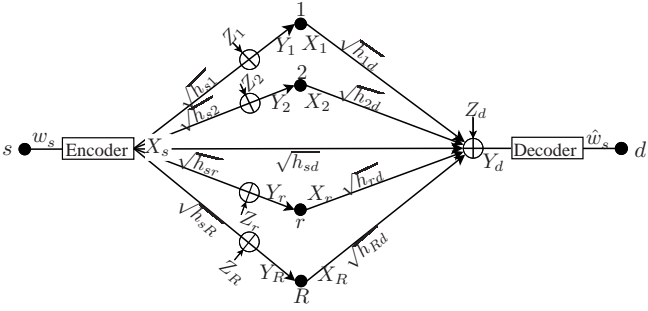}
  \end{center}
  \caption{One source, multiple parallel relays, one destination cooperative communication network.}
  \label{fig:onesmultiplerencdec}
\end{figure}

The sent message $\signal[\sourceindex]{}$ is a sequence of sub{-}messages $\signal[\sourceindex]{1},\dots,\signal[\sourceindex]{\blockindex},\dots,\signal[\sourceindex]{\blocknumber}$ uniformly drawn from a message set with size $\signalnumber$ and rate $\rate[\sourceindex]{}$ represented by
${\signalset[\sourceindex]=\{0, 1, ..., 2^{\signalnumber\rate[\sourceindex]{}}-1\}}$.
First, the transmitter encodes the message $\signal[\sourceindex]{}$ in a symbol sequence $\var[\sourceindex][1]\left(\signal[\sourceindex]{1}\right),\dots, \var[\sourceindex][\blocknumber]\left(\signal[\sourceindex]{\blocknumber}\right)$ under the constraint that $\expectation[{\var[\sourceindex]^2}]\le\powermax[\sourceindex]{}$, with $\expectation[{\var[\sourceindex]}]=0$. Each $\var[\sourceindex][\blockindex]\left(\signal[\sourceindex]{\blockindex}\right)$ represents the $\blockindex${th} random Gaussian codeword in the alphabet space $\varSpace[\sourceindex]$. In each block with index $\blockindex$, the $\var[\sourceindex][\blockindex]\left(\signal[\sourceindex]{\blockindex}\right)$ is broadcast to the destination and to all relays.
A noisy version of the transmitted signal approaches the relay $\relayindex$ as:
\begin{align}\label{eq:sRbroadcast}
\othervar[\relayindex][\blockindex] = \sqrt{\pathgain[\sourceindex]{\relayindex}}\,\var[\sourceindex][\blockindex]\left(\signal[\sourceindex]{\blockindex}\right) + \noise[\relayindex]
\end{align}
where $\noise[\relayindex]\sim\gaussian\left(0, \noisePower\right)$ is additive white Gaussian noise at the proper relay.
At the same time, each relay $\relayindex$ executes a function of $\othervarSpace[\relayindex]\to\varSpace[\relayindex]$. That is, every relay $\relayindex$ processes the received signal in the previous block index, generates a new signal
and then forward the new symbol $\var[\relayindex][\blockindex]\left(\signal[\sourceindex]{\blockindex-1}\right)$ to destination.  The information $\var[\relayindex][\blockindex]\left(\signal[\sourceindex]{\blockindex-1}\right)$ is the output of the relay $\relayindex$'s deterministic function that depends on the specific cooperative strategy whose inputs are the previous received signals:
\begin{align}\label{eq:deterministinfunc}
\var[\relayindex][\blockindex]\left(\signal[\sourceindex]{\blockindex-1}\right) = \relayfunction[\relayindex]{\blockindex}(\othervar[\relayindex][\blockindex - 1], \othervar[\relayindex][\blockindex - 2], ...,\othervar[\relayindex][1]);
\end{align}

Each relay $\relayindex$ has its own individual power constraint as $\expectation[{\var[\relayindex]^2}]\le\powermax[\relayindex]{}$. The symbol $\var[\relayindex][\blockindex]\left(\signal[\sourceindex]{\blockindex-1}\right)$ is forwarded to the destination simultaneously with broadcasting a new symbol by the source. The receive signal at the destination is given by:
\begin{align}\label{eq:Rdestreceive}
\displaystyle\othervar[\destindex][\blockindex]\!=\! \sqrt{\pathgain[\sourceindex]{\destindex}}\var[\sourceindex][\blockindex]\left(\signal[\sourceindex]{\blockindex}\right)\!+\! \sum_{\relayindex\in\relayset}{\sqrt{\pathgain[\relayindex]{\destindex}}\,\var[\relayindex][\blockindex]\left(\signal[\sourceindex]{\blockindex-1}\right)}\!+\! \noise[\destindex]
\end{align}
where $\noise[\destindex]\sim\gaussian\left(0, \noisePower\right)$. The decoding function at the destination is a Gaussian de{-}mapping function of $\othervarSpace[\destindex]\rightarrow\displaystyle\signalset[\sourceindex]$ with error probability:
\begin{align}
P_{e}^{\signalnumber}=\displaystyle{2^{-\signalnumber\rate[\sourceindex]{}}}\cdot\sum_{\signal[\sourceindex]{}}{\text{Pr }\{\destsignal[\sourceindex]{}\neq\signal[\sourceindex]{}|\signal[\sourceindex]{}\,\text{was sent}\}}
\end{align}
based on the assumption that the messages are independent, and uniformly distributed over the alphabet space $\signalset[\sourceindex]$. The minimum rate $\rate[\sourceindex]{}$ is achievable if there exists a sequence of code $\left(\signalnumber,2^{\signalnumber\rate[\sourceindex]{}}\right)$ for which $P_e^{\signalnumber\to\infty}$ is arbitrarily close to zero.

To go on with our analysis, we define the parameter $\corrsr[\sourceindex]{\relayindex}$ as the correlation between the output sequences of $\var[\sourceindex]$ and $\var[\relayindex]$ where $\relayindex\in\relayset$, i.e.  $\corrsr[\sourceindex]{\relayindex}=\displaystyle\frac{\expectation[{\var[\sourceindex]\var[\relayindex]}]}{\sqrt{\expectation[{\var[\sourceindex]^{2}}]\expectation[{\var[\relayindex]^{2}}]}}$ , and using a similar formulation, we represent the parameter
$\corrrr[\relayindex]{\otherrelay}$ as the correlation between the outputs of two relay nodes $\relayindex, \otherrelay\in\relayset$, i.e. the statistical correlation between $\var[\relayindex]$ and $\var[\otherrelay]$.

\section{Cutset upper bound}\label{sec:cutset1M1}

To compute the achievable cutset capacity of \figurename~\ref{fig:onesmultiplerencdec} cooperative network, we resort on the general representation of Aref's tool \cite[Th. 3.4]{aref-thesis}.

\begin{proposition}\label{pro:Arefcutset}
\cite[Th. 3.4]{aref-thesis} A network consisting of one transmitter, multiple relays, and one destination satisfies:
\begin{equation}\label{eq:Arefcutset1M1}
\displaystyle\upperbound[cutset]{(\sourceindex;\relayset;\destindex)} = \underset{p(\varsample[\sourceindex], \varsetsample[\relayset])}{\sup}\; \min_{\relaysubset\subseteq\relayset} \;\left\{\information\left(\var[\sourceindex]\,\varset[\relaysubset]\,;\,\othervar[\destindex]\,\othervarset[\relaysubset^C]|\varset[\relaysubset^C]\right)\right\}
\end{equation}
where maximization is subject to the power constraints defined by the network, and the channel condition $p(\othervarsample[\destindex]\,\othervarsetsample[\relayset]\,|\,\varsample[\sourceindex]\,\varsetsample[\relayset])$.\\ $\text{ }$\hfill$\blacksquare$
\end{proposition}

The aim of \propname~\ref{pro:Arefcutset} is to find the minimum data flow between the source and receiver via one of the $\sum_{\relayindex=0}^{\relaynumber}\binom{\relaynumber}{\relayindex}$ possible relay selections. There are $2^{\relaynumber}$ possibilities to select subset $\relaysubset\subset\relayset$, known as \emph{network cuts}.
In \equationname~\ref{eq:Arefcutset1M1}, suppose the mutual information term achieves its minimum with an $\setname\subseteq\relayset$, i.e.  $\upperbound[cutset]{(\sourceindex;\relayset;\destindex)}=\sup_{p(\varsample[\sourceindex], \varsetsample[\relayset])} \information\left(\var[\sourceindex]\varset[\setname]\,;\,\othervar[\destindex]\othervarset[\setname^C]|\varset[\setname^C]\right)$.
This means that the members of $\setname$ are in a sufficiently good conditions to (perfectly) receive signals from the transmitter rather than to deliver signals to the destination. On the other hand, the relays in $\setname^C$ are in a good positions to (error freely) convey signals to the destination.

The relation stated in \propname~\ref{pro:Arefcutset} was proved in a general case where the relays are physically connected to each other.
For the cooperative network scenario of \figurename~\ref{fig:onesmultiplerencdec}, where there is no cooperation among the relays, references
 \cite{delcoso-thesis, gastpar-thesis, Kramer05} represent remarkable contributions for which concern the maximum achievable capacity. In \cite[Sec. 5.4]{gastpar-thesis} and \cite[Sec. 3.2.2]{delcoso-thesis} only two broadcast $\left(\relaysubset=\emptyset\right)$ and multiple-access $\left(\relaysubset=\relayset\right)$ cuts have been considered and the power constraint was not identified for each node.
In the present work, for the scenario of \figurename~\ref{fig:onesmultiplerencdec}, first we demonstrate that $\upperbound[cutset]{(\sourceindex;\relayset;\destindex)}$ takes the maximum when the Gaussian variables $\varset[\relayset]$ are fully correlated, i.e. $\corrrr[\relayindex]{\otherrelay}=1\;\;\forall\relayindex, \otherrelay\in\relayset$. Then, we approach the upper bound capacity considering all $2^{\relaynumber}$ network cuts.

\begin{theorem}\label{th:corrrr1}
In the network of \figurename~\ref{fig:onesmultiplerencdec}, the $\upperbound[cutset]{(\sourceindex;\relayset;\destindex)}$ takes the maximum value when the Gaussian variables $\varset[\relayset]$ are fully correlated.
\end{theorem}
\begin{proof}
We expand \equationname~\ref{eq:Arefcutset1M1} as:
\begin{multline*}
\information\left(\var[\sourceindex]\,\varset[\relaysubset]\,;\,\othervar[\destindex]\,\othervarset[\relaysubset^C]|\varset[\relaysubset^C]\right)=\\ \qquad=\entropy\left(\,\othervar[\destindex]\,\othervarset[\relaysubset^C]|\varset[\relaysubset^C]\right)- \entropy\left(\,\othervar[\destindex]\,\othervarset[\relaysubset^C]|\varset[\relayset]\var[\sourceindex]\right).
\end{multline*}
In the second term, $\entropy\left(\,\othervar[\destindex]\,\othervarset[\relaysubset^C]|\varset[\relayset]\var[\sourceindex]\right)$ is not a function of $\corrrr[\relayindex]{\otherrelay}$. To prove the theorem, it is enough to show that the maximum of $\entropy\left(\,\othervar[\destindex]\,\othervarset[\relaysubset^C]|\varset[\relaysubset^C]\right)$ is achieved when ${\corrrr[\relayindex]{\otherrelay}=1}$ $\forall\,\relayindex,\otherrelay\in\relayset$, irrespective of the value $\corrsr[\sourceindex]{\relayindex}$. This is somehow obvious for Gaussian random variables because conditioning reduces entropy except the case of fully correlated given variables condition. That is $\entropy\left(A|B\right)\ge\entropy\left(A|B, C\right)$ , with equality whenever $B$ and $C$ are fully correlated, or $A$ is conditionally independent of $C$ given $B$, i.e. $A\leftrightarrow B\leftrightarrow C$.
\end{proof}

The following results can be derived from \theoremname~\ref{th:corrrr1}. First, it is reasonable if we assume the same correlation coefficient between $\var[\sourceindex]$ and every $\var[\relayindex]${s}. Thus, in the rest of this section, we suppose $\corrsr[]{}=\corrsr[\sourceindex]{\relayindex}\;\forall\,\relayindex\in\relayset$. The second result is that  $\covs\lefto[\othervar[]|\varset[\relayset]\right]=\Varop(\othervar[]|\var[\relayindex])$.


In the following, we find a general formula, considering all $2^{\relaynumber}$ cuts, for a relayed communication consisting of one source, multiple parallel relays and one destination where each node has its own individual power constraint.

\begin{theorem}\label{th:cutsetMrelays}
For one source, multiple parallel relays, one destination network, the cutset upper bound of \propname~\ref{pro:Arefcutset} is shortened to:
\begin{multline}\label{eq:cutset1M1}
\displaystyle\upperbound[cutset]{(\sourceindex;\relayset;\destindex)} = \underset{p(\varsample[\sourceindex], \varsetsample[\relayset])}{\sup} \min \left\{\min_{\relayindex\in\relayset}{\{\information\left(\var[\sourceindex]\,;\,\othervar[\destindex]\,\othervar[\relayindex]|\var[\relayindex]\right)\}};\right.\\
\information\left(\var[\sourceindex]\,\varset[\relayset]\,;\,\othervar[\destindex]\right)\bigg\}\quad
\end{multline}
\end{theorem}
\begin{proof}
The second term is the multiple{-}access capacity of all relays and the source node at the destination, i.e. $\relaysubset=\relayset$. To prove the theorem, it is enough to show that the following equality holds:
\begin{equation*}
\min_{\relaysubset\subsetneq\relayset} \;\left\{\information\left(\var[\sourceindex]\,\varset[\relaysubset]\,;\,\othervar[\destindex]\,\othervarset[\relaysubset^C]|\varset[\relaysubset^C]\right)\right\} = \,\min_{\relayindex\in\relayset}\,\{\information\left(\var[\sourceindex]\,;\,\othervar[\destindex]\,\othervar[\relayindex]|\var[\relayindex]\right)\} .
\end{equation*}
The expansion of the left{-}hand side gives:
\begin{flalign*}
&\;\displaystyle\information\left(\var[\sourceindex]\,\varset[\relaysubset]\,;\,\othervar[\destindex]\,\othervarset[\relaysubset^C]|\varset[\relaysubset^C]\right) =\\
&\qquad\qquad\underbrace{\information\left(\,\varset[\relaysubset]\,;\,\othervar[\destindex]\,\othervarset[\relaysubset^C]|\varset[\relaysubset^C]\right)}_{(1)\,=\,0} +\,\underbrace{\information\left(\var[\sourceindex]\,;\,\othervar[\destindex]\,\othervarset[\relaysubset^C]|\varset[\relayset]\right)}_{(2)}\\
&\;(1):\,\information\left(\,\varset[\relaysubset]\,;\,\othervar[\destindex]\,\othervarset[\relaysubset^C]|\varset[\relaysubset^C]\right)=\frac{1}{2}\log_{2}{\frac{\covs\lefto[\,\othervar[\destindex]\,\othervarset[\relaysubset^C]|\varset[\relaysubset^C]\right]}{\covs\lefto[\,\othervar[\destindex]\,\othervarset[\relaysubset^C]|\varset[\relayset]\right]}} =\\
&\qquad\qquad\qquad\qquad=\frac{1}{2}\log_{2}{\frac{\covs\lefto[\,\othervar[\destindex]\,\othervarset[\relaysubset^C]|\var[\relayindex\in\relaysubset^C]\right]}{\covs\lefto[\,\othervar[\destindex]\,\othervarset[\relaysubset^C]|\var[\relayindex\in\relaysubset^C]\right]}} = 0.&\\
&\;(2):\,\information\left(\var[\sourceindex]\,;\,\othervar[\destindex]\,\othervarset[\relaysubset^C]|\varset[\relayset]\right)=\,\information\left(\var[\sourceindex]\,;\,\othervar[\destindex]\,\othervarset[\relaysubset^C]|\var[\relayindex]\right) \qquad\qquad\qquad\qquad\qquad\qquad\qquad
\end{flalign*}
Until now, we demonstrated that:
\begin{equation*}
\displaystyle\min_{\relaysubset\subsetneq\relayset} \!\left\{\information\left(\var[\sourceindex]\,\varset[\relaysubset];\othervar[\destindex]\,\othervarset[\relaysubset^C]|\varset[\relaysubset^C]\right)\right\} \!=\!\min_{\relaysubset\subsetneq\relayset} \!\left\{\information\left(\var[\sourceindex];\othervar[\destindex]\,\othervarset[\relaysubset^C]|\var[\relayindex]\right)\right\}
\end{equation*}
Since $\relaysubset$ is every strictly subset of $\relayset$, it is clear to conclude that:
\begin{equation*}
\displaystyle\min_{\relaysubset\subsetneq\relayset} \;\left\{\information\left(\var[\sourceindex]\,;\,\othervar[\destindex]\,\othervarset[\relaysubset^C]|\var[\relayindex]\right)\right\}=\,\min_{\relayindex\in\relayset}\{\information\left(\var[\sourceindex]\,;\,\othervar[\destindex]\,\othervar[\relayindex]|\var[\relayindex]\right)\}.
\end{equation*}
\end{proof}

The result of \equationname~\ref{eq:cutset1M1} is the minimum value between $\relaynumber+1$ mutual information terms.
The first term of the $\upperbound[cutset]{(\sourceindex;\relayset;\destindex)}$ is the broadcast capacity from the source node to the destination and the relay with which has the minimum information flow rate. The second term is the multiple{-}access capacity of all relays and the transmitter from the destination point of view.
The novel result of \theoremname~\ref{th:cutsetMrelays} is:
\emph{the upper bound cutset capacity of a point{-}to{-}point multiple parallel relayed network is achieved either using one relay only or using all relays together.}
If all relays are in good conditions to transmit signals to the destination rather than to receive signals from the source, $\upperbound[cutset]{(\sourceindex;\relayset;\destindex)}=\upperbound[cutset]{(\sourceindex;\relayindex;\destindex)}$ wherein $\relayindex$ is the relay with which the source node achieves the smallest broadcast capacity of $\information\left(\var[\sourceindex]\,;\,\othervar[\destindex]\,\othervar[\relayindex]|\var[\relayindex]\right)$. \emph{In such a network the existence of all others relays is useless}.
On the other hand, if all relays are in good conditions to receive data from the transmitter rather than to convey signals to the destination,
the $\upperbound[cutset]{(\sourceindex;\relayset;\destindex)}$ is achieved with maximization of the multiple{-}access capacity of all relays and the source node at the destination, i.e. $\information\left(\var[\sourceindex]\,\varset[\relayset]\,;\,\othervar[\destindex]\right)$.
As can be seen, the capacity of the broadcast term of $\upperbound[cutset]{(\sourceindex;\relayset;\destindex)}$ is much less than that presented by \cite[Sec. 3.2.2]{delcoso-thesis} and \cite[Sec. 5.4]{gastpar-thesis}. This is because, \cite{delcoso-thesis} and \cite{gastpar-thesis} consider only two broadcast and multiple-access cuts rather than all $2^\relaynumber$ possible cuts.

At this point, we recall some useful equalities relevant to the equations of \figurename~\ref{fig:onesmultiplerencdec}, with $\corrrr[\relayindex]{\otherrelay}=1\;\forall\,\relayindex, \otherrelay\in\relayset$. To review the algebraic manipulation of the following formulas see \cite[Ch. 5]{Allan09}.
\begin{subequations}\label{eqs:statistics1M1}
\begin{flalign}
&\!\Varop\left(\othervar[\destindex]\right)=\pathgain[\sourceindex]{\destindex}\powermax[\sourceindex]{}+\sum_{\relayindex\in\relayset}{2\corrsr[]{}\sqrt{\pathgain[\sourceindex]{\destindex}\powermax[\sourceindex]{}\pathgain[\relayindex]{\destindex}\powermax[\relayindex]{}}}
+\nonumber\\ &\quad\qquad\qquad\qquad+\sum_{\relayindex\in\relayset}\sum_{\otherrelay\in\relayset}{\sqrt{\pathgain[\relayindex]{\destindex}\powermax[\relayindex]{}\pathgain[\otherrelay]{\destindex}\powermax[\otherrelay]{}}}\;+\noisePower\\
&\covs\lefto[\othervar[\destindex]\othervar[\relayindex]|\varset[\relayset]\right]\!=\! \covs\lefto[\othervar[\destindex]\othervar[\relayindex]|\var[\relayindex]\right]\!=\!\left(\pathgain[\sourceindex]{\destindex}\!+\! \pathgain[\sourceindex]{\relayindex}\right)\powermax[\sourceindex]{}\!\left(1\!-\!{\corrsr[]{}}^2\right)\!+\!\noisePower\\
&\covs\lefto[\othervar[\destindex]\othervar[\relayindex]|\varset[\relayset]\var[\sourceindex]\right]=\covs\lefto[\othervar[\destindex]\othervar[\relayindex]|\var[\relayindex]\var[\sourceindex]\right]=\noisePower\\
&\covs\lefto[\othervar[\destindex]|\varset[\relayset]\var[\sourceindex]\right]=\noisePower
\end{flalign}
\end{subequations}


In the following we introduce the capacity of the cooperative framework depicted by \figurename~\ref{fig:onesmultiplerencdec} with Gaussian channels.
\begin{theorem}\label{th:AWGNonesmultipleRcutset}
The AWGN cutset upper bound of the capacity one source, multiple parallel relays is:
\begin{multline*}
\!\upperbound[cutset]{(\sourceindex;\relayset;\destindex)}\!=\!\underset{-1\le\corrsr[]{}\le1}{\sup}\!\min\! \left\{\!\min_{\relayindex\in\relayset}\!\left\{\!\capacitySymbol\!\left(\frac{\covs\lefto[\othervar[\destindex]\othervar[\relayindex]|\var[\relayindex]\right]\!-\!\noisePower}{\noisePower}\right)\!\right\} \;\right.\\
,\;\capacitySymbol\left(\frac{\Varop\left(\othervar[\destindex]{}\right)-\noisePower}{\noisePower}\right)\bigg\}
\end{multline*}
\end{theorem}
which is calculated using \equationsname~\ref{eqs:statistics1M1}.
\begin{proof}
\begin{flalign*}
\;&\displaystyle\information(\var[\sourceindex]\;;\; \othervar[\destindex]{}\,\othervar[\relayindex]{} | \var[\relayindex])=\frac{1}{2}\log_2{\frac{\covs\lefto[\othervar[\destindex]\othervar[\relayindex]|\var[\relayindex]\right]}{\covs\lefto[\othervar[\destindex]\othervar[\relayindex]|\var[\relayindex]\var[\sourceindex]\right]}}
=&\\
&\quad\qquad=\frac{1}{2}\log_2{\frac{\covs\lefto[\othervar[\destindex]\othervar[\relayindex]|\var[\relayindex]\right]}{\noisePower}}=
\capacitySymbol\left(\frac{\covs\lefto[\othervar[\destindex]\othervar[\relayindex]|\var[\relayindex]\right]-\noisePower}{\noisePower}\right);\\
&\displaystyle\information(\var[\sourceindex]\,\varset[\relayset]\;;\;\othervar[\destindex]{}) = \frac{1}{2}\log_2{\frac{\Varop\left(\othervar[\destindex]{}\right)}{\covs\lefto[\othervar[\destindex]|\var[\sourceindex]\,\varset[\relayset]\right]}} =\\
&\qquad\qquad\qquad=\frac{1}{2}\log_2{\frac{\Varop\left(\othervar[\destindex]{}\right)}{\noisePower}}
=\capacitySymbol\left(\frac{\Varop\left(\othervar[\destindex]{}\right)-\noisePower}{\noisePower}\right).\;\;
\end{flalign*}
\end{proof}

If all relays are much closer to the destination than source node, or equivalently $\SNR[\sourceindex]{\relayindex}\!<\!\SNR[\relayindex]{\destindex}$ $\forall\,\relayindex\!\in\!\relayset$, then $\upperbound[cutset]{(\sourceindex;\relayset;\destindex)}\!=\!\upperbound[cutset]{(\sourceindex;\relayindex;\destindex)}$ wherein $\relayindex$ is the relay with the weakest $\SNR[\sourceindex]{\relayindex}$ channel. This means that the others relays make only ``crowd" and increase the overall power consumption. In such a relay assisted network the supremum is achieved with $\corrsr[]{}=0$. On the other hand, if all relays are located in the contrary positions, i.e. they are much closer to the transmitter than destination, $\upperbound[cutset]{(\sourceindex;\relayset;\destindex)}$ is achieved with the multiple{-}access term (the second term) and the supremum is achieved when $\var[\sourceindex]$ and $\varset[\relayset]$ are fully correlated, i.e. $\corrsr[]{}=1$.

In a multiple parallel relay network wherein all relays are well located to correctly receive signals from the transmitter, i.e. $\SNR[\sourceindex]{\relayindex}>\SNR[\relayindex]{\destindex}$, the $\upperbound[cutset]{(\sourceindex;\relayset;\destindex)}$ dominates the parallel channels capacity that is $\capacitySymbol(\SNR[\sourceindex]{\destindex}+\sum_{\relayindex\in\relayset}{\SNR[\relayindex]{\destindex}})$. This is because Gaussian variables $\varset[\relayset]$ are fully correlated.
In such a reliable network, adding a new relay, always close to the source node, expands the upper bound of the cutset capacity region, up to the number of relays until which the multiple{-}access capacity does not exceed that broadcast (the first term).

Another interesting result of the \theoremname~\ref{th:AWGNonesmultipleRcutset} is:
\emph{adding a new relay does not always increase the cutset upper bound capacity}. In a one source{-}multiple parallel relays{-}one destination system, locating a new relay very close to the destination may decrease the upper bound capacity, and adding a new relay very close to the transmitter may do not change upper bound capacity.

\section{Amplify and forward technique}\label{sec:AFTech1M1}

In a reliable network applying AF technique, the relays do not have any code space. There exist two Gaussian codebooks: $\varSpace[\sourceindex]$ at the source's encoder, and $\othervarSpace[\destindex]$ at the destination's decoder. Suppose a transmit message $\signal[\sourceindex]{}$ which is a sequence of $\blocknumber$ sub{-}messages $\signal[\sourceindex]{1},...,\signal[\sourceindex]{\blockindex},...,\signal[\sourceindex]{\blocknumber}$ and each sub{-}message $\signal[\sourceindex]{\blockindex}$ is uniformly drawn from $\signalset[\sourceindex]=\{0, 1, ..., 2^{\signalnumber\rate[\sourceindex]{}}\!-\!1\}$. Each sub-message $\signal[\sourceindex]{\blockindex}$ is separately encoded to $\var[\sourceindex][\blockindex]\left(\signal[\sourceindex]{\blockindex}\right)$ under the constraint that $\expectation[{\var[\sourceindex]^2}]\le\powermax[\sourceindex]{}$. In each block index $\blockindex$, each relay \relayindex scales the amplitude of the analog observed signal $\othervar[\relayindex][\blockindex-1]$ as:
\begin{equation}\label{eq:AFMrelays}
\begin{split}
\var[\relayindex][\blockindex]\left(\signal[\sourceindex]{\blockindex-1}\right)&=\AFconstant[\relayindex][\blockindex].\othervar[\relayindex][\blockindex-1]\\
&=\AFconstant[\relayindex][\blockindex].\left(\sqrt{\pathgain[\sourceindex]{\relayindex}}\,\var[\sourceindex][\blockindex-1]\left(\signal[\sourceindex]{\blockindex-1}\right)+\noise[\relayindex]\right)
\end{split}
\end{equation}
wherein the amplification factor $\AFconstant[\relayindex]$ is chosen so as to satisfy the proper relay's power constraint. Each relay node has its own power constraint as ${\expectation[{\var[\relayindex]^2}]\le\powermax[\relayindex]{}}$. We assume all channels are slow time{-}varying. So that we can assume $\AFconstant[\relayindex][\blockindex]=\AFconstant[\relayindex]$ and
\begin{align}\label{eq:AFfactorlimit1M1}
|\AFconstant[\relayindex]|^{2}\le\frac{\powermax[\relayindex]{}}{\,\noisePower + \pathgain[\sourceindex]{\relayindex}\powermax[\sourceindex]{}}
\end{align}

As can be seen, if $\noisePower + \pathgain[\sourceindex]{\relayindex}\powermax[\sourceindex]{}\gg\powermax[\relayindex]{}$ the (amplification of the) relay $\relayindex$ is useless. Replacing \equationname~\ref{eq:AFMrelays} into \eqref{eq:Rdestreceive}, the received signal at the destination is:
\begin{multline}\label{eq:AFdestreceiveMrelays}
\hspace{-0.3cm}\othervar[\destindex][\blockindex]\!=\! \sqrt{\pathgain[\sourceindex]{\destindex}}\,\var[\sourceindex][\blockindex]\left(\signal[\sourceindex]{\blockindex}\right)+
\sum_{\relayindex\in\relayset}{|\AFconstant[\relayindex]|\sqrt{\pathgain[\sourceindex]{\relayindex}\pathgain[\relayindex]{\destindex}}\var[\sourceindex][\blockindex\!-\!1]\left(\signal[\sourceindex]{\blockindex\!-\!1}\right)}
\\+\sum_{\relayindex\in\relayset}{|\AFconstant[\relayindex]|\sqrt{\pathgain[\relayindex]{\destindex}}\noise[\relayindex]} + \noise[\destindex]
\end{multline}

The relay nodes do not regenerate any new code, and consequently the complexity of this scheme is low. Since every relay node amplifies whatever it receives, including noise, it is mainly useful in high SNR environments.
Increasing the amplification factor $\AFconstant[]$ increases the noise of inter symbol interference (ISI) at the destination, and also significantly increase the noise of inter carrier interference (ICI) where there is only one antenna at the destination. The relay should thus transmit with an appropriate power where the network is able to adjust the power of the relay. By \equationname~\ref{eq:AFdestreceiveMrelays}, the maximum data rate of the AF scheme is formulated as:
\begin{multline}
\displaystyle\upperbound[AF]{(\sourceindex;\relayset;\destindex)} =  \\ \capacitySymbol\left(\frac{\left(\sqrt{\pathgain[\sourceindex]{\destindex}}+\sum_{\relayindex\in\relayset}{|\AFconstant[\relayindex]|\sqrt{\pathgain[\sourceindex]{\relayindex}\pathgain[\relayindex]{\destindex}}}\right)^2.\powermax[\sourceindex]{}}{\left(1+\sum_{\relayindex\in\relayset}{|\AFconstant[\relayindex]|^{2}\pathgain[\relayindex]{\destindex}}\right)\noisePower}\right)
\end{multline}


Reference \cite[p. 46]{LiuBook09} demonstrates that under the condition
 $|\AFconstant[\relayindex]| \le \SNR[\sourceindex]{\relayindex}$,
the $\displaystyle\upperbound[AF]{(\sourceindex;\relayset;\destindex)}$ outperforms the capacity of the maximal ratio combining (MRC) technique that is:
\begin{equation}\label{eq:MRC1M1}
\displaystyle\upperbound[MRC]{(\sourceindex;\relayset;\destindex)}=\capacitySymbol\left(\SNR[\sourceindex]{\destindex}+\sum_{\relayindex\in\relayset}{\frac{\SNR[\sourceindex]{\relayindex}\,\SNR[\relayindex]{\destindex}}{\SNR[\sourceindex]{\relayindex}+\SNR[\relayindex]{\destindex}}}\right)
\end{equation}

\section{Case study}\label{sec:casestudy1M1}

\begin{figure}
  \begin{center}
    \psfrag{s}[c][c][0.9]{$\sourceindex$}
    \psfrag{r1}[c][c][0.9]{$\relayindex$}
    \psfrag{r2}[c][c][0.9]{$\relayindex$}
    \psfrag{d}[c][c][0.9]{$\destindex$}
    \psfrag{d1}[c][c][1]{$d_{\sourceindex\relayindex}$}
    \psfrag{h}[c][c][1]{$d_{\relayindex}$}
    \psfrag{1}[c][b][1]{$d_{\sourceindex\destindex}$}
    \includegraphics[width=0.55\columnwidth]{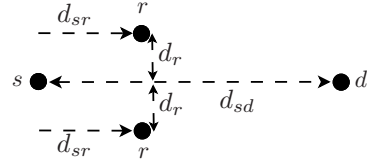}
  \end{center}
  \caption{Point{-}to{-}point two relayed communication network scenario.}
  \label{fig:Net_R1M1}
\end{figure}

Here, we illustrate the various outer region bounds of a point{-}to{-}point multiple relayed communication with Gaussian channels wherein the transmitter $\sourceindex$, two relays $\relayindex$, and the destination $\destindex$ are located as depicted in \figurename~\ref{fig:Net_R1M1}. We assume a vertical equidistance of $d_{\relayindex}$ between two relays and the $\sourceindex\to\destindex$ direct{-}link.
The path condition values ${\pathgain[\sourceindex]{\relayindex}=\pathgain[\relayindex]{\destindex}=\pathgain[\sourceindex]{\destindex}=1}$ are scaled with respect to Notation~2 in \cite{Shams13}.
First, we experiment a high SNR environment and suppose that the source and destination are located at a distance of $d_{\sourceindex\destindex}=1\m$, and the relays are located in vertical distances of $d_{\relayindex}=0.1\m$ and they are simultaneously and horizontally moving from $d_{\sourceindex\relayindex}=-0.5\m$ to $d_{\sourceindex\relayindex}=1.5\m$. \figurename~\ref{fig:high_1M1} plots various data rates for $\powermax[\sourceindex]{}=\powermax[\relayindex]{}=100\mW$, and $\noisePower=1\muW$. The curve labeled AF shows the outer region of AF strategy with the largest possible scaling factor $\AFconstant[\relayindex]$ in \equationname~\ref{eq:AFfactorlimit1M1}.
We studied a point{-}to{-}point communication relayed by only one intermediate node in \cite{Shams13} and here we use the same parameters.
The curve labeled $\corrsr[]{}$ plots a particular value of the correlation coefficient which is the same as \cite[Fig. 4]{Shams13}.
As the relays moves toward the destination the received signal at the relay becomes weaker and this significantly reduces AF data rate.
\begin{figure}
  \begin{center}
    \psfrag{xaxis}[c][c][0.9]{$d_{\sourceindex\relayindex} \left[\m\right]$}
    \psfrag{yaxis}[c][c][0.9]{Rate $\left[\bpsHz\right]$}
    \includegraphics[width=0.8\columnwidth]{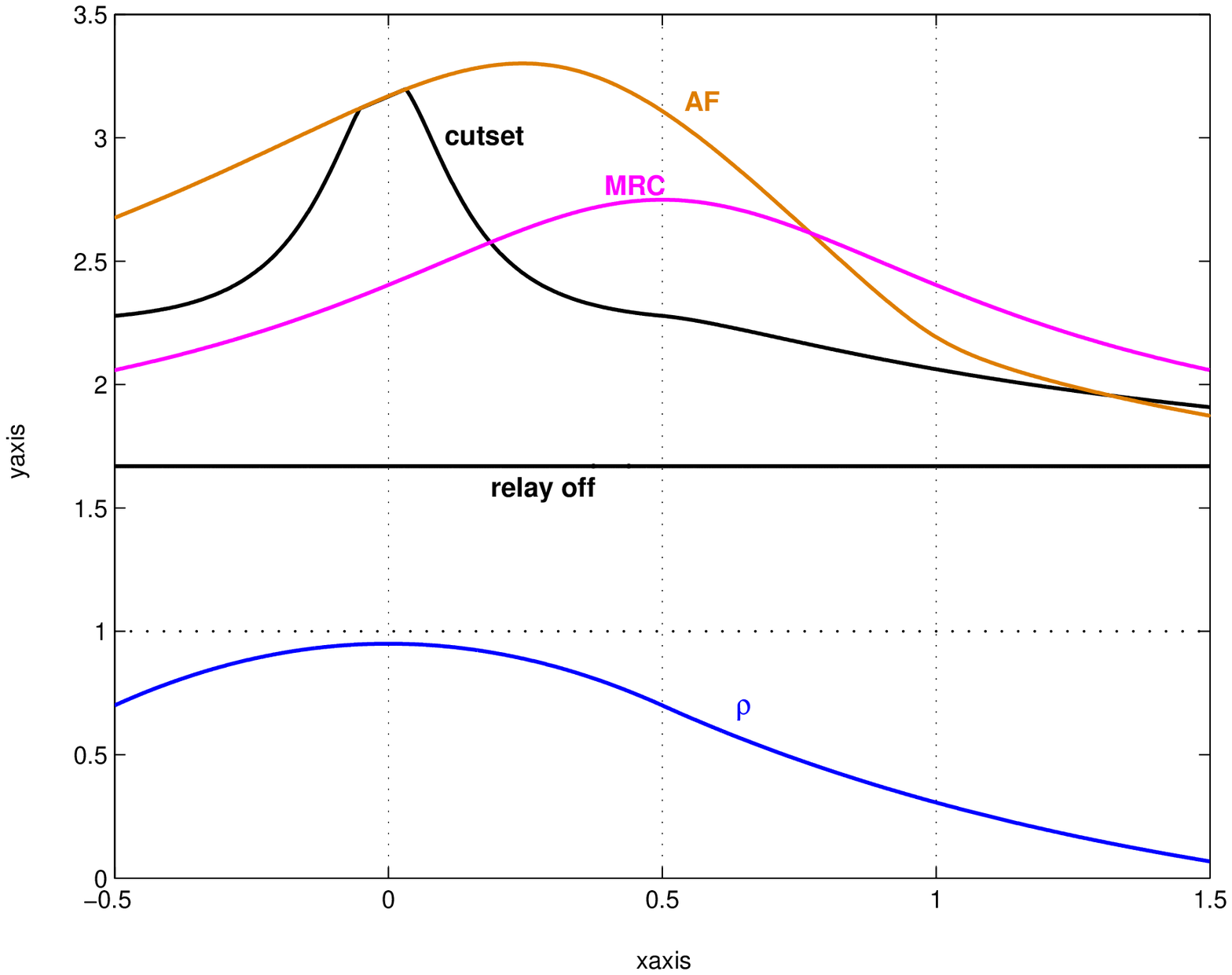}
  \end{center}
  \caption{Rates for two relays with $\powermax[\sourceindex]{}=\powermax[\relayindex]{}=100\mW$, $\noisePower=1\muW$,  $d_{\sourceindex\destindex}=1\m$, and $d_{\relayindex}=0.1\m$.}
  \label{fig:high_1M1}
\end{figure}
The comparison of \figurename~\ref{fig:high_1M1} to \cite[Fig. 4]{Shams13} reveals that when the relays are close to the transmitter, the data rates of two{-}relays network exhibits almost $25\%$ higher performance. On the contrary, when the relays are close to the destination, data rate is equal to that of the one{-}relay network in \cite[Fig. 4]{Shams13}. This is because, in such a situation the maximum cutset data rate of coding techniques is achieved by one relay only.

Now, we consider a point{-}to{-}point relayed connection in a low SNR regime. The transmitter and the destination are placed at a distance of $d_{\sourceindex\destindex}=500\m$, and the relays are simultaneously and horizontally moving in a range of ${d_{\sourceindex\relayindex}=-100\div600\m}$ with vertical distance of $d_{\relayindex}=10\m$.
\figurename~\ref{fig:low_1M1} plots various data rates for $\powermax[\sourceindex]{}=\powermax[\relayindex]{}=100\mW$, $\noisePower=1\muW$, and the same $\AFconstant[\relayindex]$ as the previous simulation.

\begin{figure}
  \begin{center}
    \psfrag{xaxis}[c][c][0.9]{$d_{\sourceindex\relayindex} \left[\m\right]$}
    \psfrag{y1axis}[c][c][0.9]{Rate $\left[\bpsHz\right]$}
    \psfrag{y2axis}[c][c][0.9]{\color{blue}Correlation coefficient $\corrsr[]{}$}
    \includegraphics[width=0.85\columnwidth]{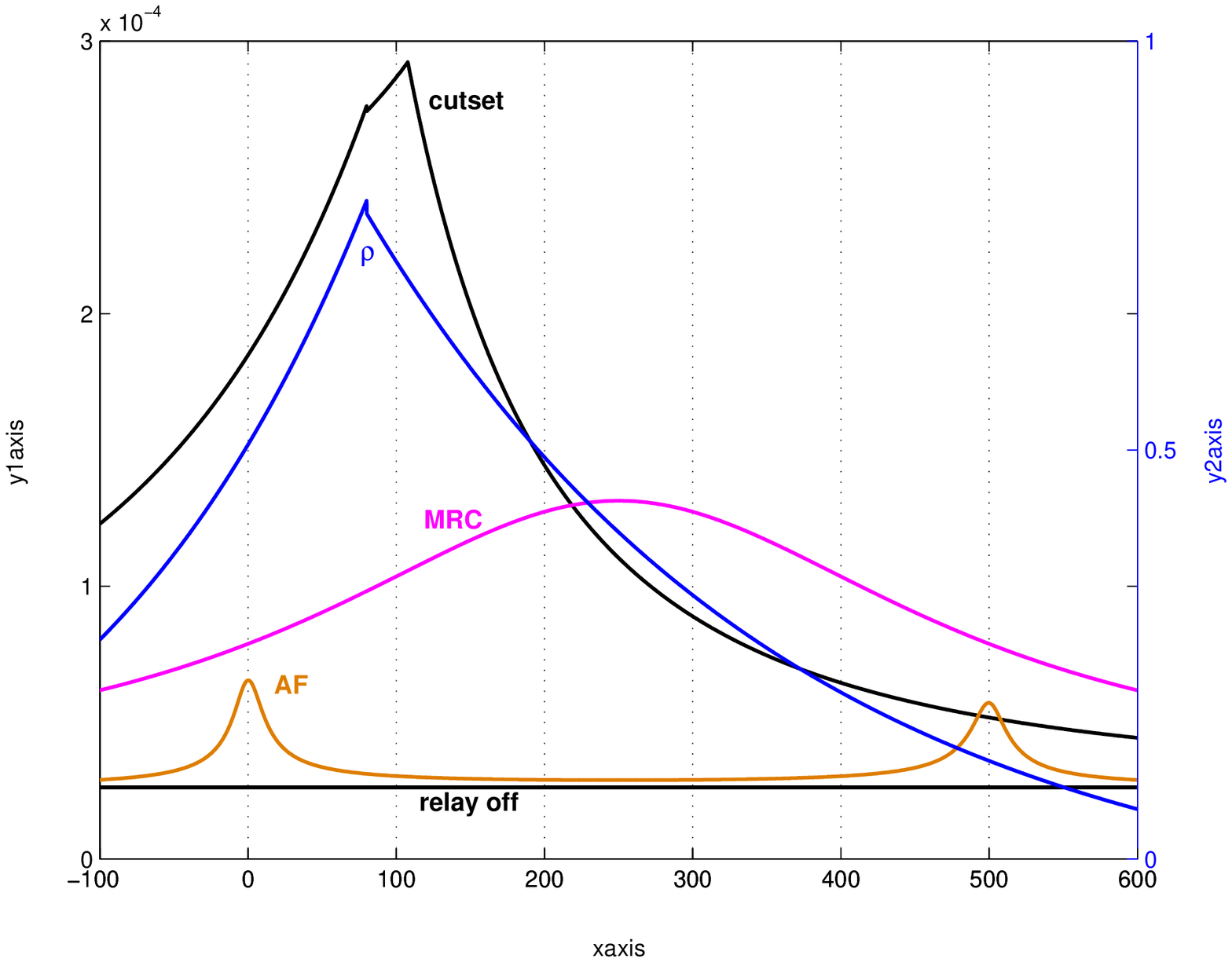}
  \end{center}
  \caption{Rates for one relay with $\powermax[\sourceindex]{}=\powermax[\relayindex]{}=100\mW$, $\noisePower=1\muW$, ${d_{\sourceindex\destindex}=500\m}$, and $d_{\relayindex}=10\m$.}
  \label{fig:low_1M1}
\end{figure}

We draw a different experimental function for correlation value which is the curve labeled $\corrsr[]{}$, and that is the same as the correlation function in \cite[Fig. 5]{Shams13}.
From \figurename~\ref{fig:low_1M1}, it is clearly derived that, like in the one{-}relay network, the AF technique is not quite useful in a low SNR network, whereas  MRC technique performs much better than AF.
Our experiments in a given scenario with different parameters reveals that reducing the amplification factor does not increase the AF data rate.
The comparison between \figurename~\ref{fig:low_1M1} and \cite[Fig. 5]{Shams13} shows that the data rates of coding techniques in two{-}relays network is almost $50\%$ higher than those of one{-}relay network. This means that adding a new relay in a low SNR network is much more useful than that in a high SNR network.
Like one{-}relay case in \cite[Fig. 5]{Shams13}, in a low SNR scene, the cutset technique shows significant higher rate than the relay off mode.

\section{Summary}\label{sec:conclusion1M1}

We studied point{-}to{-}point communications aided by multiple parallel relays with full{-}duplex signaling and AWGN channels.
We focused on finding the maximum achievable capacity applying cutset and AF relay strategies.
First, we showed that, in coding strategy, the maximum data rate is achieved when the output signals of the relays are fully correlated.
The first novel result is that the maximum source to destination cutset flow rate is approached by either only one relay or all relays together.

Like in the single{-}relay networks, the performance of multiple parallel relays channels basically depends upon both the strategy of the relays and their positions.
If all relays are located so as to perfectly receive signals from the source, the maximum capacity is limited by the weakest source{-}to{-}relay channel data rate. In this situation, the others relays can be turned off.
Placing a relay very far from the source node can even decrease the source to destination direct{-}link capacity.
On the other hand, when all relays are located so as to perfectly deliver signals to the destination, the maximum capacity is equal to the multiple{-}access capacity of all relays at the destination. In such a network, adding a new relay may expand the outer region of capacity.

\setlength{\IEEEilabelindent}{2\IEEEiedmathlabelsep}
\IEEEusemathlabelsep

\bstctlcite{mybibfile:BSTcontrol}
\bibliography{IEEEabrv,mybibfile}

\end{document}